\documentclass{article}
\usepackage{graphicx}
\usepackage{latexsym}
\usepackage{amssymb}
\usepackage{amsmath}
\usepackage{amsthm}
\usepackage{thmtools}
\usepackage{forloop}
\usepackage[paper=letterpaper,margin=1in]{geometry}
\usepackage[ruled,vlined,linesnumbered]{algorithm2e}
\usepackage{paralist}
\usepackage{xcolor}
\usepackage{scalerel}
\usepackage{pgfplots}
\pgfplotsset{compat=1.18}
\usepackage{hyperref}

\newtheorem{theorem}{Theorem}[section]
\newtheorem{lemma}[theorem]{Lemma}

\newtheorem{corollary}[theorem]{Corollary}
\newtheorem{definition}{Definition}[section]

\newtheorem{observation}[theorem]{Observation}

\newcommand{\defcal}[1]{\expandafter\newcommand\csname c#1\endcsname{{\mathcal{#1}}}}
\newcommand{\defbb}[1]{\expandafter\newcommand\csname b#1\endcsname{{\mathbb{#1}}}}
\newcommand{\defvec}[1]{\expandafter\newcommand\csname v#1\endcsname{{\mathbf{#1}}}}
\newcounter{calBbCounter}
\forLoop{1}{26}{calBbCounter}{
    \edef\letter{\alph{calBbCounter}}
		\edef\Letter{\Alph{calBbCounter}}
    \expandafter\defcal\Letter
		\expandafter\defbb\Letter
		\expandafter\defvec\letter
}

\newcommand{\eps}{\varepsilon}
\newcommand{\nnR}{{\bR_{\geq 0}}}

\newcommand{\characteristic}{{\mathbf{1}}}

\newcommand{\RSet}{{\texttt{R}}}
\newcommand{\email}[1]{{\href{mailto:#1}{#1}}}

\title{Stronger Hardness for Submodular Maximization Subject to a Matroid Constraint}

\author{Moran Feldman\thanks{Computer Science Department, University of Haifa, Israel. Part of this work was done while the author was visiting Queen Mary University of London. E-mail: \email{moranfe@cs.haifa.ac.il}} \and Alan Kuhnle\thanks{Department of Computer Science \& Engineering, Texas A\&M University. E-mail: \email{kuhnle@tamu.edu}}}

\begin{document}

\maketitle

\begin{abstract}
Maximizing a submodular function subject to a matroid constraint is a cornerstone problem in combinatorial optimization. The state-of-the-art algorithm for this problem obtains $0.401$-approximation~\cite{buchbinder2024constrained}, and the state-of-the-art hardness result shows that no polynomial time algorithm can obtain better than $0.478$-approximation for this problem~\cite{oveisgharan2011submodular}. In this work, we present the first improvement in $15$ years for the hardness result, showing that no polynomial time algorithm can obtain better than $8/17 \approx 0.471$-approximation, even for the special case of a cardinality or (simplified) partition matroid constraint.

\medskip


\end{abstract}

\section{Introduction}

Maximization of a submodular function subject to a matroid constraint is a cornerstone problem in combinatorial optimization, important both for its many applications and for the valuable techniques developed as part of the study of this problem. An important special case is the maximization of a \emph{monotone} submodular function subject to a matroid constraint. Already in $1978$, Fisher et al.~\cite{fisher1978analysis} showed that the greedy algorithm guarantees $1/2$-approximation for this special case, and Nemhauser et al.~\cite{nemhauser1978analysis} showed that the approximation ratio of the greedy algorithm improves to $1-1/e \approx 0.632$ when the matroid constraint happens to be a cardinality constraint. The last approximation ratio was shown to be tight by Nemhauser \& Wolsey~\cite{nemhauser1978best}, but this left a gap for general matroid constraints between the approximation ratio of $1/2$ due to~\cite{fisher1978analysis} and the hardness of $1-1/e$ due to~\cite{nemhauser1978best}. This gap was finally closed roughly $30$ years later by Calinescu et al.~\cite{calinescu2011maximizing}, who described the celebrated Continuous Greedy algorithm, which guarantees $(1-1/e)$-approximation for maximizing a monotone submodular function subject to a general matroid constraint.

Understanding the maximization of general (not necessarily monotone) submodular functions subject to a matroid constraint turned out to be more difficult. The first algorithm for this task was presented in $2010$ by Lee et al.~\cite{lee2010maximizing} and guaranteed $1/4$-approximation. This result was improved in a long series of works~\cite{buchbinder2019constrained,chekuri2014submodular,ene2016constrained, feldman2011unified,oveisgharan2011submodular,vondrak2013symmetry}, and the current state-of-the-art is a recent $0.401$-approximation algorithm due to Buchbinder \& Feldman~\cite{buchbinder2024constrained}. A few inapproximability results are also known for the problem. The $1-1/e$ inapproximability of Nemhauser \& Wolsey~\cite{nemhauser1978best} for the monotone case naturally applies also to general submodular functions, and this was improved by Feige et al.~\cite{feige2011maximizing}, who showed a hardness of $1/2$ for maximizing a general submodular function even in the absence of a constraint (which is a special case of a matroid constraint). Later, Vondr{\'{a}}k~\cite{vondrak2013symmetry} described a new framework for proving inapproximability results for submodular maximization problems, which he named \emph{the symmetry gap} framework. Oveis Gharan \& Vondr{\'{a}}k~\cite{oveisgharan2011submodular} used the symmetry gap framework to prove an inapproximability of $2e^{-1/2}(1 - e^{-1/2}) \approx 0.477$ for maximizing a general submodular function subject to a matroid constraint.\footnote{Oveis Gharan \& Vondr{\'{a}}k~\cite{oveisgharan2011submodular} states the hardness as the numerical approximation $0.478$. The exact expression was determined by~\cite{feldman2025bicriteria}.} The hardness of~\cite{oveisgharan2011submodular} applies even when the matroid constraint is restricted be a partition matroid constraint, and Qi~\cite{qi2024maximizing} described a variant of it that works even when the matroid constraint is restricted to be a cardinality constraint.

Despite many attempts to improve over it, the hardness of Oveis Gharan \& Vondr{\'{a}}k~\cite{oveisgharan2011submodular} remained the state-of-the-art for $15$ years. Feldman \& Kuhnle~\cite{feldman2025bicriteria} even showed that the Measured Continuous Greedy algorithm of~\cite{feldman2011unified} can guarantee a matching approximation ratio for a large family of instances; which was taken as an evidence for why the hardness of Oveis Gharan \& Vondr{\'{a}}k~\cite{oveisgharan2011submodular} is likely to be tight. Nevertheless, in this work we show that this is not the case, and prove the following theorem.

\begin{theorem} \label{thm:hardness}
For every constant $\eps > 0$, no sub-exponential time algorithm can guarantee $(8/17 + \eps)$-approxi\-mation for the problem of maximizing a non-negative submodular function subject to a matroid constraint. Furthermore, this is the case even if the matroid constraint is either restricted to be a (simplified) partition matroid or restricted to be a cardinality constraint.
\end{theorem}

We note that $8/17$ is approximately $0.471$. Furthermore, the assumption of non-negativity in Theorem~\ref{thm:hardness} is standard, and is done in all the above stated works. Without this assumption, one cannot obtain any constant multiplicative guarantee for the problem.

\paragraph{Remark.} After the appearance of the first version of this note on arXiv, Siyuan Chen brought to our attention that a very recent concurrent work of Vaneet Aggarwal~\cite{aggarwal2026geometry} proves a hardness of 0.470438681 for a related setting of DR-submodular function maximization. Siyuan Chen observed that this hardness result of Vaneet Aggarwal can be transferred to the problems of maximization of a non-negative submodular set function over a cardinality or matroid constraint, and thus, slightly improves over the hardness of 8/17 proved by Theorem~\ref{thm:hardness} for these constraints.

\paragraph{Paper Structure.} Section~\ref{sec:preliminaries} gives the necessary technical background, including the definitions necessary for the formal statement of the problem we consider and the details of the symmetry gap framework. Then, the hardness stated by Theorem~\ref{thm:hardness} for general matroid constraints is proved in Section~\ref{sec:symmetric_instance}. Section~\ref{sec:specific_classes} shows how to get the same hardness for more specific classes of matroid constraints.

\section{Preliminaries} \label{sec:preliminaries}

In this section, we present known definitions and results. Specifically, in Section~\ref{ssc:problem_definition}, we give the definitions necessary for formally stating the problem we consider. Then, Section~\ref{ssc:symmetry_gap} formally presents the symmetry gap framework. We conclude with Section~\ref{ssc:additional_notation}, which states some additional notation that we use.

\subsection{Problem Related Definitions} \label{ssc:problem_definition}

\paragraph{Set Functions.}

A \emph{set function} is a function $f\colon 2^N \to \bR$ assigning a numerical value to every subset of some ground set $N$. We say that $f$ is \emph{non-negative} if $f(S) \geq 0$ for every $S \subseteq N$. We say that $f$ is \emph{monotone} if $f(S) \leq f(T)$ for every two sets $S \subseteq T \subseteq N$. Finally, we say that $f$ is \emph{submodular} if $f(S) + f(T) \geq f(S \cup T) + f(S \cap T)$ for every two sets $S, T \subseteq N$. It follows from this definition that submodular functions are closed under addition and multiplication by a non-negative scalar. In other words, if $f(S)$ and $g(S)$ are two submodular functions and $\alpha \geq 0$ is a scalar, then the function $\alpha \cdot f(S) + g(S)$ is also submodular. An equivalent definition of submodularity is that $f$ is submodular if $f(S \cup \{u\}) - f(S) \geq f(T \cup \{u\}) - f(T)$ for every two sets $S \subseteq T \subseteq N$ and element $u \in N \setminus T$.

The multilinear extension of a set function $f \colon 2^N \to \bR$ is a function $F\colon [0, 1]^N \to \bR$ defined by
\[
    F(\vx)
    \triangleq
    \sum_{S \subseteq N} \bigg[f(S) \cdot \prod_{u \in S} x_u \cdot \mspace{-9mu} \prod_{u \in N \setminus S} \mspace{-9mu} (1 - x_u)\bigg]
    =
    \bE[f(\RSet(\vx))]
    \qquad
    \forall\; \vx \in [0, 1]^N
    \enspace,
\]
where $\RSet(\vx)$ is a random set independently containing every element $u \in N$ with probability $x_u$.

\paragraph{Matroids.}

A matroid $\cM$ is defined as a pair $(N, \cI)$, where $N$ is a ground set and $\cI$ is a non-empty collection of subsets of $N$ obeying two properties.
\begin{itemize}
    \item Down-closedness: For every two sets $S \subseteq T \subseteq N$, if $T \in \cI$ then $S \in \cI$.
    \item Exchange: For every two sets $S, T \in \cI$, if $|S| < |T|$, then there exists an element $u \in T \setminus S$ such that $S \cup \{u\} \in \cI$.
\end{itemize}
Following linear algebra motivations, it is customary to refer to the sets of $\cI$ as the \emph{independent} sets of $\cM$. Under the corresponding matroid constraint, a set is feasible if it is independent in $\cM$. 

A (generalized) \emph{partition matroid} $\cM$ is defined by a partition of the ground set $N$ into disjoint sets $N_1, N_2, \dotsc, N_h$ and non-negative integers $k_1, k_2, \dotsc, k_h$. A set $S \subseteq N$ is independent in $\cM$ if $|N_i \cap S| \leq k_i$ for every $i \in [h]$. One can verify that partition matroids are indeed matroids. Furthermore, $\cM$ is called a \emph{simplified partition matroid} if $k_i = 1$ for all $i \in [h]$, and it is called a \emph{uniform matroid} if there is only a single part in the partition (i.e., $h = 1$). The matroid constraint corresponding to a uniform matroid allows any subset of $N$ of size at most $k$ for some value $k$. Such constraints are also known as \emph{cardinality constraints}.

\paragraph{Oracles.}

We are interested in the problem of maximizing a submodular function $f \colon 2^N \to \bR$ subject to a matroid constraint $\cM = (N, \cI)$. Algorithms for this problem are interesting when they run in time that is polynomial in the size of the ground set $N$, and therefore, it is important to make the formal input size polynomial in $|N|$. Unfortunately, the description of $f$ and $\cI$ might be exponential in $|N|$. Thus, it is customary to assume that $f$ is not part of the formal input to the algorithm, and instead the algorithm can access it through a \emph{value oracle}. Given a set $S \subseteq N$, the value oracle returns $f(S)$.

The situation regarding $\cI$ is somewhat more involved. When $\cM$ is a uniform or partition matroid, $\cI$ has a compact representation, and thus, is considered part of the formal input. However, when $\cM$ is a general matroid constraint, it is customary to assume that the algorithm has access to $\cI$ only through an \emph{independence oracle} that given a set $S \subseteq N$ answers whether $S \in \cI$.

\subsection{Symmetry Gap} \label{ssc:symmetry_gap}

The main result of the symmetry gap framework of~\cite{vondrak2013symmetry} is given by Theorem~\ref{thm:symmetry_gap} below. However, before getting to this theorem, we first need to state some definitions from~\cite{vondrak2013symmetry}. The first definition states what it means for a problem instance to be symmetric. In all the definitions of this section, a constraint $\cF$ is formally modeled as the collection of all the sets that are feasible under it.

\begin{definition}[Strong symmetry] \label{def:strong_symmetry}
Consider a non-negative submodular function $f\colon 2^N \to \nnR$ and a collection $\cF \subseteq 2^N$ of feasible sets. The problem $\max\{f(S) \mid S \in \cF\}$ is strongly symmetric with respect to a group of permutations $\cG$ on $N$, if (1) $f(S) = f(\sigma(S))$ for all $S \subseteq N$ and $\sigma \in \cG$, and (2) $S \in \cF \iff S' \in \cF$ whenever $\bE_{\sigma \in \cG} [\characteristic_{\sigma(S)}] = \bE_{\sigma \in \cG} [\characteristic_{\sigma(S')}]$, where $\bE_{\sigma \in \cG}$ represents the expectation over picking $\sigma$ uniformly at random out of $\cG$.
\end{definition}

The symmetry gap of a symmetrical instance is the ratio between the optimal (possibly fractional) symmetric solution for the instance and the true (possibly fractional) optimal solution. Following is a formal definition of the symmetry gap.
\begin{definition}[Symmetry gap]
Consider a non-negative submodular function $f\colon 2^N \to \nnR$ and a collection $\cF \subseteq 2^N$ of feasible sets. Let $F(x)$ be the multilinear extension of $f$ and $P(\cF) \subseteq [0, 1]^N$ be the convex hull of $\{\characteristic_S \mid S \in \cF\}$. Then, if the problem $\max\{f(S) \mid S \in \cF\}$ is strongly symmetric with respect to a group $\cG$ of permutation, then its symmetry gap is defined as
\[
	\frac{\max\{F(\bar{\vx}) \mid \vx \in P(\cF)\}}{\max\{F(\vx) \mid \vx \in P(\cF)\}}
	\enspace,
\]
where $\bar{\vx} \triangleq \bE_{\sigma \in \cG}[\sigma(\vx)]$.
\end{definition}

The symmetry gap framework shows that the existence of a symmetrical instance with a constraint $\cF$ translates into a hardness for the family of constraints that can be obtained as refinements of $\cF$ (in some sense). Following is the formal definition of these refinements.

\begin{definition}[Refinement] \label{def:refinement}
Consider a set $\cF \subseteq 2^N$, and let $X$ be some set. We say that $\tilde{\cF} \subseteq 2^{N \times X}$ is a refinement
of $\cF$ if
\[
 \tilde{F}
	=
	\Big\{S \subseteq N \times X ~\Big|~ \text{$\vx \in P(\cF)$, where $x_u = \tfrac{|S \cap (\{u\} \times X)|}{|X|}$ for every $u \in N$} \Big\}
	\enspace.
\]
\end{definition}

We are now ready to state the main result of the symmetry gap framework.

\begin{theorem}[Theorem~1.8 of~\cite{vondrak2013symmetry}] \label{thm:symmetry_gap}
Consider a non-negative submodular function $f\colon 2^N \to \nnR$ and a collection $\cF \subseteq 2^N$ of feasible sets such that the problem $\max\{f(S) \mid S \in \cF\}$ is strongly symmetric with respect to a group $\cG$ of permutations and has a symmetry gap $\gamma$. Let $\cC$ be the class of the problems $\max\{ \tilde{f}(S) \mid S \in \tilde{\cF}\}$ in which $\tilde{f}$ is a non-negative submodular function and $\tilde{\cF}$ is a
refinement of $\cF$. Then, for every constant $\eps > 0$, any (even randomized) $(1 + \eps)\gamma$-approximation algorithm for the class $\cC$ would require exponentially many value oracle queries to $\tilde{f}(S)$. If $f$ is monotone, then this hardness result holds even if the class $\cC$ is restricted to problems with a monotone objective function $\tilde{f}$.
\end{theorem}

\subsection{Additional Notation} \label{ssc:additional_notation}

Given a condition $C$, we denote by $\characteristic[C]$ an indicator for the event that this condition holds. Additionally, given a subset $S$ of a ground set $N$, we denote by $\characteristic_S$ the characteristic vector of $S$, i.e., a vector in $\{0, 1\}^N$ that assigns the value one to the coordinates corresponding to the elements of $S$.

\section{Symmetric Instance} \label{sec:symmetric_instance}

In this section, we describe the symmetric instance underlying our hardness results, and then derive Theorem~\ref{thm:hardness} for general matroid constraints by applying Theorem~\ref{thm:symmetry_gap} to this symmetric instance. Our symmetric instance is parameterized by three values: an integer $\ell \geq 1$ and two positive real values $w_1$ and $w_2$ whose sum is at most $1$. The ground set $N$ of the instance is the union of three disjoint sets $A \triangleq \{a_i \mid i \in [\ell]\}$, $B \triangleq \{b_i \mid i \in [\ell]\}$ and $C =  \{c_{i, 1}, c_{i, 2} \mid i \in [\ell]\}$. The constraint $\cF$ of the instance is
\[
    \cF
    \triangleq
    \{S \subseteq N \mid |S \cap  A| \leq 1\}
    \enspace.
\]
Notice that this constraint is a simplified partition matroid constraint defined by a partition whose parts contain a single element each, except for one part that contains all the elements of $A$.

To define the objective function of the symmetric instance, it is useful to define first for every two non-empty disjoint sets $X, Y \subseteq N$ the function
\[
    g_{X, Y}(S)
    \triangleq
    \characteristic[X \cap S \neq \varnothing] \cdot \characteristic[Y \cap S \neq Y]
    \qquad
    \forall\; S \subseteq N
    \enspace.
\]
Intuitively, $g$ is the cut function of a hyperedge whose foot and head include the elements of $X$ and $Y$, respectively. The objective function of the symmetric instance is then given for every $S \subseteq N$ by
\[
    f(S)
    \triangleq
    \sum_{i = 1}^\ell \bigg[\ell(1 - w_1 - w_2) \cdot g_{\{a_i\}, \{b_i\}}(S) + w_1 \cdot g_{\{a_i, b_i\}, \{c_{i, 1}, c_{i, 2}\}}(S) + w_2 \cdot \sum_{j = 1}^2 g_{\{c_{i, j}\}, \{c_{i, 3 - j}\}}(S)\bigg]
    \enspace.
\]

\begin{observation} \label{obs:f_properties}
The function $f$ is non-negative and submodular.
\end{observation}
\begin{proof}
The non-negativity of $f$ is immediate, therefore, we concentrate on proving that $f$ is submodular. By the closure of submodular functions to addition and multiplication by a non-negative scalar, it suffices to argue that $g_{X, Y} (\cdot)$ is a submodular function for every two non-empty disjoint sets $X, Y \subseteq N$. Intuitively, this is true since $g_{X, Y}$ is a cut function. However, to prove it formally, we need to show that for arbitrary two sets $S, T \subseteq N$, it holds that $g_{X, Y}(S \cup T) + g_{X, Y}(S \cap T) \leq g_{X, Y}(S) + g_{X, Y}(T)$. We show that by considering three cases.

The first case is when $S \cap X = \varnothing$. In this case, $g_{X, Y}(S) = 0$, and thus,
\begin{align*}
    g_{X, Y}(S \cup T) +{}& g_{X, Y}(S \cap T)\\
    ={} &
    \characteristic[X \cap (S \cup T) \neq \varnothing] \cdot \characteristic[Y \cap (S \cup T) \neq Y] + \characteristic[X \cap S \cap T \neq \varnothing] \cdot \characteristic[Y \cap S \cap T \neq Y]\\
    ={} &
    \characteristic[X \cap T \neq \varnothing] \cdot \characteristic[Y \cap (S \cup T) \neq Y] + \characteristic[\varnothing \neq \varnothing] \cdot \characteristic[Y \cap S \cap T \neq Y]\\
    \leq{} &
    \characteristic[X \cap T \neq \varnothing] \cdot \characteristic[Y \cap T \neq Y] + 0 \cdot \characteristic[Y \cap S \cap T \neq Y]
    =
    g_{X, Y}(T) + g_{X, Y}(S)
    \enspace.
\end{align*}
The second case is when $S \cap Y = Y$. In this case, again $g_{X, Y}(S) = 0$, and thus,
\begin{align*}
    g_{X, Y}(S \cup T) +{}& g_{X, Y}(S \cap T)\\
    ={} &
    \characteristic[X \cap (S \cup T) \neq \varnothing] \cdot \characteristic[Y \cap (S \cup T) \neq Y] + \characteristic[X \cap S \cap T \neq \varnothing] \cdot \characteristic[Y \cap S \cap T \neq Y]\\
    ={} &
    \characteristic[X \cap (S \cup T) \neq \varnothing] \cdot \characteristic[Y \neq Y] + \characteristic[X \cap S \cap T \neq \varnothing] \cdot \characteristic[Y \cap T \neq Y]\\
    \leq{} &
    \characteristic[X \cap (S \cup T) \neq \varnothing] \cdot 0 + \characteristic[X \cap T \neq \varnothing] \cdot \characteristic[Y \cap T \neq Y]
    =
    g_{X, Y}(S) + g_{X, Y}(T)
    \enspace.
\end{align*}
Finally, consider the case that $S \cap X \neq \varnothing$ and $S \cap Y \neq Y$. In this case $g_{X, Y}(S) = 1$, and thus,
\begin{align*}
    g_{X, Y}(S \cup T) +{}& g_{X, Y}(S \cap T)\\
    ={} &
    \characteristic[X \cap (S \cup T) \neq \varnothing] \cdot \characteristic[Y \cap (S \cup T) \neq Y] + \characteristic[X \cap S \cap T \neq \varnothing] \cdot \characteristic[Y \cap S \cap T \neq Y]\\
    ={} &
    1 \cdot \characteristic[Y \cap (S \cup T) \neq Y] + \characteristic[X \cap S \cap T \neq \varnothing] \cdot 1\\
    \leq{} &
    \characteristic[Y \cap T \neq Y] + \characteristic[X \cap T \neq \varnothing]
    \leq
    1 + g_{X, Y}(T)
    =
    g_{X, Y}(S) + g_{X, Y}(T)
    \enspace.
    \qedhere
\end{align*}
\end{proof}

Our next goal is to bound the symmetry gap of $\max\{f(S) \mid S \in \cF\}$. We begin by showing that there is a valuable integral solution for this problem.
\begin{observation} \label{obs:opt_value}
There is a set $O \in \cF$ such that $f(O) \geq \ell$.
\end{observation}
\begin{proof}
Let us define $O = \{a_1, c_{1, 1}\} \cup \{b_i, c_{i, 1} \mid 2 \leq i \leq \ell\}$. Notice that $O$ indeed belongs to $\cF$ since its intersection with $A$ contains only the element $a_1$. Furthermore,
\[
    g_{\{a_1\}, \{b_1\}}(O) = 1 \enspace,
    \qquad
    g_{\{a_i, b_i\}, \{c_{i, 1}, c_{i, 2}\}}(O) = 1 \quad \forall\; i \in [\ell]
\]
and
\[
    g_{\{c_{i, 1}\}, \{c_{i, 2}\}}(O) = 1
    \quad
    \forall\; i \in [\ell]
    \enspace.
\]
Therefore,
\[
    f(O)
    \geq
    \ell(1 - w_1 - w_2) + \ell w_1
    + \ell w_2
    =
    \ell
    \enspace.
    \qedhere
\]
\end{proof}

Next, let us show that any symmetric solution for $\max\{f(S) \mid S \in \cF\}$ has only low value. One can verify that $\max\{f(S) \mid S \in \cF\}$ is strongly symmetric with respect to the group $\cG$ of permutations formed by taking the closure of the following permutations over the ground set $N$.
\begin{itemize}
    \item Any permutation reordering the indices $i$. Formally, for every permutation $\pi$ over $[\ell]$, there is such a permutation $\sigma_\pi$ over $N$ given by
    \[
        \sigma_\pi(u)
        =
        \begin{cases}
            a_{\pi(i)} & \text{if $u = a_i$} \enspace,\\
            b_{\pi(i)} & \text{if $u = b_i$} \enspace,\\
            c_{\pi(i), j} & \text{if $u = c_{i, j}$}
        \end{cases}
        \qquad
        \forall\; u \in N
        \enspace.
    \]
    \item The permutation $\sigma_{\text{s}}$ that swaps the pairs of elements in $C$ with each other. More formally,
    \[
        \sigma_{\text{s}}(u)
        =
        \begin{cases}
            a_{i} & \text{if $u = a_i$} \enspace,\\
            b_{i} & \text{if $u = b_i$} \enspace,\\
            c_{i, 3 - j} & \text{if $u = c_{i, j}$}
        \end{cases}
        \qquad
        \forall\; u \in N
        \enspace.
    \]
\end{itemize}

Recall that $P(\cF)$ is the convex hull of $\{\characteristic_S \mid S \in \cF\}$.
\begin{lemma} \label{lem:symmetric_value}
For every vector $\vx \in P(\cF)$, let $\bar{\vx} = \bE_{\sigma \in \cG}[\sigma(\vx)]$. Then,
\[
    F(\bar{\vx})
    \leq
    \ell \cdot \max\bigg\{1 - w_1 - w_2/2, \frac{(w_1 + w_2)^2}{w_1 + 2w_2}\bigg\} + 1
    \enspace.
\]
\end{lemma}
\begin{proof}
Notice that the fact that $\sigma_\pi \in \cG$ for every permutation $\pi$ of $[\ell]$ implies that $\bar{x}_{a_i}$ has the same value for all $i \in [\ell]$. Let us denote by $a$ this common value. Similarly, there must be common values $b$, $c_1$ and $c_2$ such that $\bar{x}_{b_i} = b$, $\bar{x}_{c_{i, 1}} = c_1$ and $\bar{x}_{c_{i, 2}} = c_2$ for all $i \in [\ell]$. Furthermore, the membership of $\sigma_s$ in $\cG$ implies that $c_1 = \bar{x}_{c_{1, 1}} = \bar{x}_{c_{1, 2}} = c_2$, and thus, we can denote by $c$ the common value of $c_1$ and $c_2$.

Given the above observations and notation, we can write
\begin{align} \label{eq:basic_upper_bound}
    F(\bar{\vx})
    ={} &
    \ell^2(1 - w_1 - w_2) \cdot a(1 - b) + \ell w_1 \cdot (a + b - ab)(1 - c^2) + 2\ell w_2c(1 - c)\\\nonumber
    \leq{} &
    \ell(1 - w_1 - w_2)(1 - b) + w_1 \cdot (1 + b(\ell - 1))(1 - c^2) + 2\ell w_2c(1 - c)
    \enspace,
\end{align}
where the inequality holds since the membership of $\bar{\vx}$ in $P(\cF)$ implies that $a\ell = \sum_{i = 1}^\ell \bar{x}_{a_i} \leq 1$.
When $b = 0$, the rightmost side of Inequality~\eqref{eq:basic_upper_bound} reduces to 
\[
    \ell(1 - w_1 - w_2) + w_1(1 - c^2) + 2\ell w_2c(1 - c)
    \leq
    \ell(1 - w_1 - w_2) + 1 + \ell w_2/2
    =
    \ell(1 - w_1 - w_2/2) + 1
    \enspace,
\]
where the inequality uses the fact that $c(1 - c) \leq 1/4$ for every $c \in [0, 1]$. When $b = 1$, the rightmost side of Inequality~\eqref{eq:basic_upper_bound} reduces to
\begin{align*}
    \ell w_1 (1 - c^2) + 2\ell w_2 c(1 - c)
    \leq{} &
    \ell w_1 \bigg(1 - \frac{w_2^2}{(w_1 + 2w_2)^2}\bigg) + 2\ell w_2 \cdot \frac{w_2(w_1 + w_2)}{(w_1 + 2w_2)^2}\\
    ={} &
    \ell \cdot \bigg(w_1 + \frac{w_1w_2^2 + 2w_2^3}{(w_1 + 2w_2)^2}\bigg)
    =
    \ell \cdot \bigg(w_1 + \frac{w_2^2}{w_1 + 2w_2}\bigg)
    =
    \frac{\ell(w_1 + w_2)^2}{w_1 + 2w_2}
    \enspace,
\end{align*}
where the inequality holds since the derivative of $\ell w_1 (1 - c^2) + 2\ell w_2 c(1 - c)$ with respect to $c$ is
\[
    2\ell \cdot [-w_1 c + w_2(1 - 2c)]
    =
    2\ell \cdot [w_2 - c(w_1 + 2w_2)]
    \enspace.
\]
The lemma now follows by noticing that the rightmost side of Inequality~\eqref{eq:basic_upper_bound} is a linear function of $b$, and thus, there must exist an assignment that maximizes it and assigns to $b$ one of the two extreme values $0$ or $1$.
\end{proof}

\begin{corollary} \label{cor:symmetry_gap}
The symmetry gap of the problem $\max\{f(S) \mid S \in \cF\}$ is at most $\max\big\{1 - w_1 - w_2/2,\allowbreak \frac{(w_1 + w_2)^2}{w_1 + 2w_2}\big\} + \ell^{-1}$. In particular, for $w_1 = w_2 = 6/17$, the above symmetry gap is at most $8/17 + \ell^{-1}$.
\end{corollary}
\begin{proof}
By Observation~\ref{obs:opt_value} and Lemma~\ref{lem:symmetric_value}, the symmetry gap of $\{f(S) \mid S \in \cF\}$ is
\begin{align*}
    \frac{\max\{F(\bar{\vx}) \mid \vx \in P(\cF)\}}{\max\{F(\vx) \mid \vx \in P(\cF)\}}
    \leq{} &
    \frac{\ell \cdot \max\big\{1 - w_1 - w_2/2, \frac{(w_1 + w_2)^2}{w_1 + 2w_2}\big\} + 1}{f(O)}\\
    \leq{} &
    \max\bigg\{1 - w_1 - w_2/2, \frac{(w_1 + w_2)^2}{w_1 + 2w_2}\bigg\} + \ell^{-1}
    \enspace.
    \qedhere
\end{align*}
\end{proof}

\noindent \textbf{Remark.} The values assigned by Corollary~\ref{cor:symmetry_gap} to the parameters $w_1$ and $w_2$ are the ones minimizing the symmetry gap guarantee. See Appendix~\ref{app:determining_weights} for an explanation how these values were found.

Plugging the guarantee of Corollary~\ref{cor:symmetry_gap} into Theorem~\ref{thm:symmetry_gap}, and observing that every sub-exponential time algorithm necessarily makes only a sub-exponential number of value oracle queries, we get the following lemma.
\begin{lemma}
For every constant $\eps' > 0$, no sub-exponential time algorithm can guarantee $(1 + \eps')(8/17 + \ell^{-1})$-approximation for the problem of maximizing a non-negative submodular function subject to a constraint $\tilde{\cF}$ that is a refinement of $\cF$.
\end{lemma}

The last lemma implies the result of Theorem~\ref{thm:hardness} for general matroid constraints by setting $\eps' = \eps/3$, $\ell = \lceil 3/\eps\rceil$ and observing that every refinement of $\cF$ is a (generalized) partition matroid.

\section{Hardness for Special Classes of Matroid Constraints} \label{sec:specific_classes}

In this section, our goal is to prove Theorem~\ref{thm:hardness} for simplified partition matroid constraints and cardinality constraints. Each one of the following subsections prove it for one of these constraint types.

\subsection{Hardness for Simplified Partition Matroid Constraints}

Proving Theorem~\ref{thm:hardness} for simplified partition matroid constraints requires us to replace the concept of refinements due to~\cite{vondrak2013symmetry} with a slightly different concept, termed \emph{parallel refinement}. Intuitively, under this kind of refinement, we create multiple copies of the ground set, and a set is feasible under the refinement if it is feasible under the original constraint when restricted to each one of these copies. The idea of parallel refinements can be traced back to the proof of~\cite{feldman2011unified} for the hardness of the submodular welfare problem. 

\begin{definition}[Parallel Refinement]
Consider a set $\cF \subseteq 2^N$, and let $X$ be some set. We say that $\tilde{\cF} \subseteq 2^{N \times X}$ is a parallel refinement
of $\cF$ if
\[
    \tilde{F}
	=
	\Big\{S \subseteq N \times X ~\Big|~ \forall_{x \in X} \{u \mid (u, x) \in S\} \in \cF\} \Big\}
	\enspace.
\]
\end{definition}

The proof of Theorem~\ref{thm:symmetry_gap} uses some properties of refinements. The next two lemmata show that parallel refinements also have these properties. The first property is that parallel refinements preserve symmetry in the sense that if $\cF$ is closed under a group $\cG$ of permutations, then a parallel refinement of $\cF$ is closed under applications of permutations from $\cG$ to copies of the ground set.

\begin{lemma} \label{lem:parallel_symmetry}
Assume the problem $\max\{f(S) \mid S \in \cF\}$ is strongly symmetric with respect to the group of permutations $\cG$, and $\tilde{\cF}$ be the parallel refinement of $\cF$ obtained using some set $X$. Then, for every element $x \in X$, permutation $\sigma \in \cG$ and set $S \subseteq N \times X$,
\[
    S \in \tilde{\cF}
    \Leftrightarrow
    \{(u, x') \in S \mid x' \neq x\} \cup \{(\sigma(u), x) \mid (u, x) \in S\} \in \tilde{\cF}
    \enspace.
\]
\end{lemma}
\begin{proof}
Let us denote $\sigma(S, x) \triangleq \{(u, x') \in S \mid x' \neq x\} \cup \{(\sigma(u), x) \mid (u, x) \in S\}$.
By the definition of parallel refinement, to prove that the lemma, it suffices to show that, for every $x' \in X$, we have
\[
    \{u \mid (u, x') \in S\} \in \cF
    \Leftrightarrow
    \{u \mid (u, x') \in \sigma(S, x)\} \in \cF
    \enspace.
\]
When $x' \neq x$, this is true since $\{u \mid (u, x') \in S\} = \{u \mid (u, x') \in \sigma(S, x)\}$. When $x' = x$, the same equivalence is still true because
\[
    \{u \mid (u, x) \in \sigma(S, x)\}
    =
    \sigma(\{u \mid (u, x) \in S\})
\]
and $\cF$ is closed under permutations from $\cG$.
\end{proof}

\noindent \textbf{Remark.} The proof of Lemma~\ref{lem:parallel_symmetry} only uses the fact that $\cF$ is symmetric in the sense that it is closed under permutations of $\cG$. The definition of strong symmetry requires $\cF$ to be symmetric in a stronger sense, which is redundant for the purposes of Lemma~\ref{lem:parallel_symmetry} and Theorem~\ref{thm:paralll_hardness} below.

The next lemma shows that parallel refinements preserve the value of the optimal solution.

\begin{lemma} \label{lem:parallel_value_conserving}
Fix a constant $\eps > 0$, a non-negative set function $f\colon 2^\cN \to \nnR$ and a constraint $\cF \subseteq 2^N$. Let $F$ be the multilinear extension of $f$. Then, there exists a large enough value $n_0$ such that in every parallel refinement $\tilde{\cF}$ of $\cF$ obtained using a set $X$ of size at least $n_0$, there is a set $S \in \tilde{\cF}$ obeying
$F(\xi(S)) + \eps \geq \max\{F(\vx) \mid \vx \in P(\cF)\}$, where
    \[
        \xi_u(S)
        =
        \tfrac{1}{|X|}|\{(u, x) \in S \mid x \in X\}|
        \qquad
        \forall\; u \in N
        \enspace.
    \]
\end{lemma}
\begin{proof}
Let $\vx^*$ be a vector in $P(\cF)$ maximizing $F(\vx^*)$. Since $\vx^* \in P(\cF)$, it can be represented as a convex combination of the characteristic vectors of the sets in $\cF$. In other words, there is a non-negative coefficient $\alpha_T$ for every $T \in \cF$ such that the coefficients of all sets add up to $1$ and
\[
    \vx^*
    =
    \sum_{T \in \cF} \alpha_T \cdot \characteristic_T
    \enspace.
\]
By rounding each coefficient $\alpha_T$ either up or down to an integer multiple of $1/|X|$, we can get new non-negative coefficient $\alpha'_T$ for every set $T \in \cF$ such that the coefficients still add up to $1$ and they also obey $|\alpha'_T - \alpha_T| < 1/|X|$ for every $T \in \cF$. Using the coefficients $\alpha'_T$, it is possible to define a function $H \colon X \to \cF$ in which every set $T \in \cF$ is the image of exactly $|X| \cdot \alpha'_T$ different elements of $X$. Consider now the set
\[
    S
    \triangleq
    \bigcup_{x \in X} (H(x) \times \{x\})
    \enspace.
\]
Clearly, this set belongs to $\tilde{\cF}$ because $H(x) \in \cF$ for every $x \in X$. Additionally, for every $u \in N$,
\begin{align} \label{eq:coordiante_diffrence}
    |\xi_u(S) - x^*_u|
    ={} &
    \bigg|\frac{1}{|X|}\sum_{x \in X} \characteristic[u \in H(x)] - \sum_{T \in \cF} \alpha_T \cdot \characteristic[u \in T]\bigg|\\\nonumber
    ={} &
    \bigg|\sum_{T \in \cF} \alpha'_T \cdot \characteristic[u \in T] - \sum_{T \in \cF} \alpha_T \cdot \characteristic[u \in T]\bigg|
    \leq
    \sum_{T \in \cF} |\alpha'_T - \alpha_T|
    \leq
    \frac{|\cF|}{|X|}
    \enspace.
\end{align}
Let us now compare
\[
    F(\vx^*)
    =
    \sum_{T \subseteq \cN}
    \bigg[ f(T) \cdot \prod_{u \in T} x^*_u \cdot \prod_{u \in N \setminus T} \mspace{-9mu} (1 - x^*_u)\bigg]
\]
with
\[
    F(\xi(S))
    =
    \sum_{T \subseteq \cN}
    \bigg[ f(T) \cdot \prod_{u \in T} \xi_u(S) \cdot \prod_{u \in N \setminus T} \mspace{-9mu} (1 - \xi_u(S))\bigg]
    \enspace.
\]
By Inequality~\eqref{eq:coordiante_diffrence}, for every set $T \subseteq \cN$ the coefficient of $f(T)$ in each one of the above expressions can differ by at most $|\cF| \cdot |N|/|X|$. Thus, the difference between $F(\xi(S))$ and $F(\vx^*) = \max\{F(\vx) \mid \vx \in P(\cF)\}$ is at most
\[
    \frac{2^{|N|} \cdot |\cF| \cdot |N| \cdot \max_{T \subseteq \cN} f(T)}{|X|}
    \leq
    \frac{2^{|N|} \cdot |\cF| \cdot |N| \cdot \max_{T \subseteq \cN} f(T)}{n_0}
    \enspace.
\]
The numerator of the last fraction is a positive constant independent of $n_0$. Thus, for a large enough $n_0$, the fraction is upper bounded by $\eps$.
\end{proof}

As explained above, Lemmata~\ref{lem:parallel_symmetry} and~\ref{lem:parallel_value_conserving} establish that parallel refinements have the two properties of standard refinements used in the proof of Theorem~1.8 of~\cite{vondrak2013symmetry}: preservation of symmetry under per-copy application of permutations from $\cG$ (Lemma~\ref{lem:parallel_symmetry}), and preservation of the optimal fractional objective value (Lemma~\ref{lem:parallel_value_conserving}). The interested reader can verify that the remainder of Vondr\'{a}k's proof carries through unchanged under this substitution, yielding the following variant of Theorem~\ref{thm:symmetry_gap}.

\begin{theorem} \label{thm:paralll_hardness}
Consider a non-negative submodular function $f\colon 2^N \to \nnR$ and a collection $\cF \subseteq 2^N$ of feasible sets such that the problem $\max\{f(S) \mid S \in \cF\}$ is strongly symmetric with respect to a group $\cG$ of permutations and has a symmetry gap $\gamma$. Let $\cC$ be the class of the problems $\max\{ \tilde{f}(S) \mid S \in \tilde{\cF}\}$ in which $\tilde{f}$ is a non-negative submodular function and $\tilde{\cF}$ is a
parallel refinement of $\cF$. Then, for every constant $\eps > 0$, any (even randomized) $(1 + \eps)\gamma$-approximation algorithm for the class $\cC$ would require exponentially many value oracle queries to $\tilde{f}(S)$. If $f$ is monotone, then this hardness result holds even if the class $\cC$ is restricted to problems with a monotone objective function $\tilde{f}$.
\end{theorem}
Consider now the constraint $\cF$ from Section~\ref{sec:symmetric_instance}. We recall that $\cF$ is a simplified partition matroid constraint. As explained in Section~\ref{sec:symmetric_instance}, standard refinement of such a constraint is a (generalized) partition matroid constraint. However, a parallel refinement of a simplified matroid constraint is another constraint of the same type. Thus, by plugging the guarantees of Corollary~\ref{cor:symmetry_gap} into Theorem~\ref{thm:paralll_hardness}, and observing again that a sub-exponential time algorithm can make only a sub-exponential number of value oracle queries, we get the following lemma.
\begin{lemma}
For every constant $\eps' > 0$, no sub-exponential time algorithm can guarantee $(1 + \eps')(8/17 + \ell^{-1})$-approximation for the problem of maximizing a non-negative submodular function subject to a simplified partition matroid constraint.
\end{lemma}

The result of Theorem~\ref{thm:hardness} for simplified partition matroid constraints now follows by setting $\eps' = \eps/3$ and $\ell = \lceil 3/\eps\rceil$ in the last lemma.

\subsection{Hardness for Cardinality Constraints}

Proving Theorem~\ref{thm:hardness} for cardinality constraints requires us to somewhat modify the symmetric instance from Section~\ref{sec:symmetric_instance}. First, we create $\ell^2$ copies from each element $a_i$. More formally, if we denote $A' = \{a_{i, j} \mid i \in [\ell], j \in [\ell^2]\}$, then the new ground set is $N' = A' \cup B \cup C$. Intuitively, we think of all the elements $a_{i, 1}, a_{i, 2}, \dotsc, a_{i, \ell^2}$ as fractions of the element $a_i$ from the original ground set. Therefore, in the new constraint, we want to allow a feasible set to contain up to $\ell^2$ such fractions, in addition to any elements of the sets $B$ and $C$, which together contain $3\ell$ elements. Since we are aiming to prove a hardness for cardinality constraints, the new constraint should be a cardinality constraint, and thus, we must allow any set of size at most $3\ell + \ell^2$. Formally, the new constraint is
\[
    \cF' = \{S \subseteq N' \mid |S| \leq 3\ell + \ell^2\}
    \enspace.
\]
To complete the definition of the modified instance, we still need to define its objective function $f'$. This is done using the above intuition that $a_{i, 1}, a_{i, 2}, \dotsc, a_{i, \ell^2}$ are fractions of the element $a_i$. Specifically, given a set $S \subseteq \cN'$, let us define the vector $\xi(S) \in [0, 1]^N$ by
\[
    \xi_u(S)
    =
    \begin{cases}
        \tfrac{1}{\ell^2}|S \cap \{a_{i, j} \mid j \in [\ell^2]\}| & \text{if $u = a_i$} \enspace, \\
        \characteristic[u \in S] & \text{if $u \in N \setminus A$}
    \end{cases}
    \qquad
    \forall\; u \in N \enspace.
\]
Then, $f'(S) = F(\xi(S))$, where $F$ is the multilinear extension of the function $f$ from Section~\ref{sec:symmetric_instance}.

\begin{lemma} \label{obs:f_prime_properties}
The function $f'$ is non-negative and submodular.
\end{lemma}
\begin{proof}
The non-negativity of $f'$ follows from the non-negativity of $F$ (which in its turn follows from the non-negativity of $f$). To see that $f'$ is also submodular, consider arbitrary two sets $S \subseteq T \subseteq N'$ and an element $u \in N' \setminus T$. Since $S \subseteq T$, $\xi(S) \leq \xi(T)$, and thus, it is possible to correlate (for analysis purposes) the distributions of $\RSet(\xi(S))$ and $\RSet(\xi(T))$ so that it always holds that $\RSet(\xi(S)) \subseteq \RSet(\xi(T))$. If $u \not \in A'$, this immediately gives us
\begin{align*}
    f'(S \cup \{u\}) - f'(S)
    ={} &
    F(\xi(S \cup \{u\})) - F(\xi(S))
    =
    \bE[f(\RSet(\xi(S \cup \{u\}))) - f(\RSet(\xi(S)))]\\
    ={} &
    \bE[f(\RSet(\xi(S)) \cup \{u\}) - f(\RSet(\xi(S)))]
    \geq
    \bE[f(\RSet(\xi(T)) \cup \{u\}) - f(\RSet(\xi(T)))]\\
    ={} &
    \bE[f(\RSet(\xi(T \cup \{u\}))) - f(\RSet(\xi(T)))]
    =
    F(\xi(T \cup \{u\})) - F(\xi(T))\\
    ={} &
    f'(T \cup \{u\}) - f'(T)
    \enspace,
\end{align*}
where the inequality follows from the submodularity of $f$ since the fact that $u \in N' \setminus (A' \cup T)$ implies that $u \not \in \RSet(\xi(T))$.

To complete the proof that $f'$ is submodular, it remains to show that the same inequality holds also when $u = a_{i, j}$ for some $i$ and $j$. Let us denote by $\vee$ and $\wedge$ the vector coordinate-wise maximum and minimum operations, respectively. Since the multilinear extension is a multilinear function (as its name suggests), we get for $u = a_{i, j}$ the inequality
\begin{align*}
    f'(S \cup \{u\}) - f'&(S)
    =
    F(\xi(S \cup \{u\})) - F(\xi(S))
    =
    F(\xi(S) + \ell^{-2} \cdot \characteristic_{\{a_i\}}) - F(\xi(S))\\
    ={} &
    \ell^{-2} \cdot [F(\xi(S) \vee \characteristic_{\{a_i\}}) - F(\xi(S) \wedge \characteristic_{N \setminus \{a_i\}}) ]
    =
    \ell^{-2} \cdot \bE[f(\RSet(\xi(S)) \cup \{a_i\}) - f(\RSet(\xi(S)) \setminus \{a_i\})]\\
    \geq{} &
    \ell^{-2} \cdot \bE[f(\RSet(\xi(T)) \cup \{a_i\}) - f(\RSet(\xi(T)) \setminus \{a_i\})]
    =
    \ell^{-2} \cdot [F(\xi(T) \vee \characteristic_{\{a_i\}}) - F(\xi(T) \wedge \characteristic_{N \setminus \{a_i\}}) ]\\
    ={} &
    F(\xi(T) + \ell^{-2} \cdot \characteristic_{\{a_i\}}) - F(\xi(T))
    =
    F(\xi(T \cup \{u\})) - F(\xi(T))
    =
    f'(T \cup \{u\}) - f'(T)
    \enspace.
    \qedhere
\end{align*}
\end{proof}

We now need to define an equivalent $O'$ to the set $O = \{a_1, c_{1, 1}\} \cup \{b_i, c_{i,1} \mid 2 \leq i \leq \ell\}$  from Section~\ref{sec:symmetric_instance}. Following to the above intuition, $O'$ is obtained from $O$ by replacing $a_1$ with the $\ell^2$ elements $\{a_{1, 1}, a_{1, 2}, \dotsc, a_{1, \ell^2}\}$. Notice that $O'$ is a feasible solution as it contains $\ell^2 + 2(\ell - 1) + 1 \leq \ell^2 + 3\ell$ elements.
\begin{observation} \label{obs:opt_value_cardinality}
It holds that $f'(O') = f(O) \geq \ell$.
\end{observation}
\begin{proof}
The equality holds since $\xi(O') = \characteristic_O$, and the inequality follows from Observation~\ref{obs:opt_value}.
\end{proof}

Let us now identify a group $\cG'$ of permutations with respect to which the instance $\max\{f'(S) \mid S \in \cF'\}$ is strongly symmetric. This is the group formed by taking the closure of the following permutations over $N'$.
\begin{itemize}
    \item Any permutation reordering the indices $i$. Formally, for every permutation $\pi$ over $[\ell]$, there is such a permutation $\sigma'_\pi$ over $N$ given by
    \[
        \sigma'_\pi(u)
        =
        \begin{cases}
            a_{\pi(i), j} & \text{if $u = a_{i, j}$} \enspace,\\
            b_{\pi(i)} & \text{if $u = b_i$} \enspace,\\
            c_{\pi(i), j} & \text{if $u = c_{i, j}$}
        \end{cases}
        \qquad
        \forall\; u \in N'
        \enspace.
    \]
    \item The permutation $\sigma'_{\text{s}}$ that swaps the pairs of elements in $C$ with each other. More formally,
    \[
        \sigma'_{\text{s}}(u)
        =
        \begin{cases}
            a_{i,j} & \text{if $u = a_{i,j}$} \enspace,\\
            b_{i} & \text{if $u = b_i$} \enspace,\\
            c_{i, 3 - j} & \text{if $u = c_{i, j}$}
        \end{cases}
        \qquad
        \forall\; u \in N'
        \enspace.
    \]
\end{itemize}

Notice that the permutations defining $\cG'$ are very similar to the ones defining $\cG$. This similarity underlies the following observation.

\begin{observation} \label{obs:symmetry_xi_swap}
Assume the definition of $\xi$ is extended to vectors $\vx \in [0, 1]^{N'}$ in the natural way, i.e.,
\[
    \xi_u(\vx)
    =
    \begin{cases}
        \tfrac{1}{\ell^2} \sum_{j = 1}^{\ell^2} x_{a_{i, j}}  & \text{if $u = a_i$} \enspace, \\
        x_u & \text{if $u \in N \setminus A$}
    \end{cases}
    \qquad
    \forall\; u \in N \enspace.
\]
Then, for every vector $\vx \in [0, 1]^{N'}$, it holds that
\begin{itemize}
    \item $\xi(\sigma'_\pi(\vx)) = \sigma_\pi(\xi(\vx))$ for every permutation $\pi$ over $[\ell]$, and
    \item $\xi(\sigma'_{\text{s}}(\vx)) = \sigma_{\text{s}}(\xi(\vx))$.
\end{itemize}
Furthermore, these equalities imply $\xi(\bE_{\sigma \in \cG'}[\sigma(\vx)]) = \bE_{\sigma \in \cG}[\sigma(\xi(\vx))]$.
\end{observation}

Using the last observation, we can prove the following bound.

\begin{lemma} \label{lem:symmetric_value_cardinality}
Let $F'$ be the multi-linear extension of $f'$, and for every vector $\vx \in P(\cF')$, let $\bar{\vx} = \bE_{\sigma \in \cG'}[\sigma(\vx)]$. Then,
\[
    F'(\bar{\vx})
    \leq
    (\ell + 3) \cdot \max\bigg\{1 - w_1 - w_2/2, \frac{(w_1 + w_2)^2}{w_1 + 2w_2}\bigg\} + 1 + 3/\ell
    \enspace.
\]
\end{lemma}
\begin{proof}
Observe first that
\begin{align*}
    F'(\bar{\vx})
    ={} &
    \bE[f'(\RSet(\bar{\vx}))]
    =
    \bE[F(\xi(\RSet(\bar{\vx})))]
    =
    \bE[f(\RSet(\xi(\RSet(\bar{\vx}))))]\\
    ={} &
    \bE[f(\RSet(\xi(\bar{\vx})))]
    =
    F(\xi(\bar{\vx}))
    =
    F(\xi(\bE_{\sigma \in \cG'}[\sigma(\vx)]))
    =
    F(\bE_{\sigma \in \cG}[\sigma(\xi(\vx)]))
    \enspace,
\end{align*}
where the fourth equality holds because both $\RSet(\xi(\RSet(\bar{x})))$ and $\RSet(\xi(\bar{\vx}))$ are product distributions, and for every $i \in [\ell]$,
\[
    \Pr[a_i \in \RSet(\xi(\RSet(\bar{\vx}))]
    =
    \frac{1}{\ell^2} \cdot \sum_{j = 1}^{\ell^2} \bar{x}_{i, j}
    =
    \Pr[a_i \in \RSet(\xi(\bar{\vx})]
    \enspace.
\]

One cannot use Lemma~\ref{lem:symmetric_value} to bound the rightmost side of the last inequality because the vector $\xi(\vx)$ might not belong to $P(\cF)$ since the value of $\sum_{i \in [\ell]} \xi_{a_i}(\vx)$ is
\[
    \frac{\sum_{i = 1}^\ell \sum_{j = 1}^{\ell^2} x_{i, j}}{\ell^2}
    \leq
    \frac{\ell^2 + 3\ell}{\ell^2}
    =
    1 + \frac{3}{\ell}
    \enspace,
\]
which may be larger than $1$. However, this calculation shows that $\frac{1}{1 + 3/\ell} \cdot \xi(\vx)$ is guaranteed to belong to $\cP(\cF)$. Therefore,
\begin{align*}
    F'(\bar{\vx})
    ={} &
    F(\bE_{\sigma \in \cG}[\sigma(\xi(\vx)]))
    \leq
    (1 + 3/\ell) \cdot F\bigg(\frac{1}{1 + 3/\ell} \cdot \bE_{\sigma \in \cG}[\sigma(\xi(\vx)])\bigg)\\
    ={} &
    (1 + 3/\ell) \cdot F\bigg(\bE_{\sigma \in \cG}\bigg[\sigma\bigg(\frac{1}{1 + 3/\ell} \cdot  \xi(\vx)\bigg]\bigg)\bigg)\\
    \leq{} &
    (1 + 3/\ell) \cdot \bigg\{\ell \cdot \max\bigg\{1 - w_1 - w_2/2, \frac{(w_1 + w_2)^2}{w_1 + 2w_2}\bigg\} + 1\bigg\}\\
    ={} &
    (\ell + 3) \cdot \max\bigg\{1 - w_1 - w_2/2, \frac{(w_1 + w_2)^2}{w_1 + 2w_2}\bigg\} + 1 + 3/\ell
    \enspace,
\end{align*}
where the first inequality holds because $F$, as the multilinear extension of a submodular function, is concave along any non-negative direction~\cite{calinescu2011maximizing}; and the second inequality follows from Lemma~\ref{lem:symmetric_value} because $\frac{1}{1 + 3/\ell} \cdot \xi(\vx) \in \cP(\cF)$.
\end{proof}

\begin{corollary} \label{cor:symmetry_gap_cardinality}
The symmetry gap of the problem $\max\{f'(S) \mid S \in \cF'\}$ is at most $(1 + 3/\ell) \cdot [\max\big\{1 - w_1 - w_2/2, \frac{(w_1 + w_2)^2}{w_1 + 2w_2}\big\} + \ell^{-1}]$. In particular, for $w_1 = w_2 = 6/17$, the above symmetry gap is at most $8/17 + O(\ell^{-1})$.
\end{corollary}
\begin{proof}
By Observation~\ref{obs:opt_value_cardinality} and Lemma~\ref{lem:symmetric_value_cardinality}, the symmetry gap of $\{f'(S) \mid S \in \cF'\}$ is
\begin{align*}
    \frac{\max\{F'(\bar{\vx}) \mid \vx \in P(\cF')\}}{\max\{F'(\vx) \mid \vx \in P(\cF')\}}
    \leq{} &
    \frac{(\ell + 3) \cdot \max\big\{1 - w_1 - w_2/2, \frac{(w_1 + w_2)^2}{w_1 + 2w_2}\big\} + 1 + 3/\ell}{f'(O')}\\
    \leq{} &
    (1 + 3/\ell) \cdot \bigg[\max\bigg\{1 - w_1 - w_2/2, \frac{(w_1 + w_2)^2}{w_1 + 2w_2}\bigg\} + \ell^{-1}\bigg]
    \enspace.
    \qedhere
\end{align*}
\end{proof}

Plugging the guarantee of Corollary~\ref{cor:symmetry_gap_cardinality} into Theorem~\ref{thm:symmetry_gap}, we get the following lemma.
\begin{lemma}
For every constant $\eps' > 0$, no sub-exponential time algorithm can guarantee $(1 + \eps')(8/17 + O(\ell^{-1}))$-approximation for the problem of maximizing a non-negative submodular function subject to a constraint $\tilde{\cF}$ that is a refinement of $\cF'$.
\end{lemma}

The last lemma implies the result of Theorem~\ref{thm:hardness} for cardinality constraints by observing that every refinement of $\cF'$ is a cardinality constraint, setting $\eps' = \eps/2$ and choosing $\ell$ to be large enough function of $\eps$.

\appendix

\section{Determining the Parameters \texorpdfstring{$w_1$}{w1} and \texorpdfstring{$w_2$}{w2}} \label{app:determining_weights}

Corollary~\ref{cor:symmetry_gap} assigns values for $w_1$ and $w_2$. These values were chosen to minimize the symmetry gap guarantee. In this appendix, we explain how these values were found. Let us denote $\alpha \triangleq 1 - w_1 - w_2/2 \in (0, 1)$. Then, up to the $\ell^{-1}$ error term, the symmetry gap guaranteed by Corollary~\ref{cor:symmetry_gap} is
\begin{equation} \label{eq:symmetry_gap_simplified}
    \max\bigg\{\alpha, \frac{(1 - \alpha + w_2/2)^2}{1 - \alpha + 3w_2/2}\bigg\}
    \enspace.
\end{equation}
Notice that the derivative with respect to $w_2$ of the second argument of the maximum operation is
\begin{align*}
    &
    \frac{(1 - \alpha + w_2/2)(1 - \alpha + 3w_2/2) - (1 - \alpha + w_2/2)^2 \cdot 3/2}{(1 - \alpha + 3w_2/2)^2}\\
    ={} &
    \frac{(1 - \alpha + w_2/2)}{2(1 - \alpha + 3w_2/2)^2} \cdot [2(1 - \alpha + 3w_2/2) - 3(1 - \alpha + w_2/2)]
    =
    \frac{(1 - \alpha + w_2/2)}{2(1 - \alpha + 3w_2/2)^2} \cdot (-1 + \alpha + 3w_2/2)
    \enspace.
\end{align*}
Since $\alpha \in (0, 1)$ and $w_2 > 0$, the fraction on the rightmost side is always positive. Thus, the sign of the derivative is equal to the sign of $-1 + \alpha + 3w_2/2$, which implies that (for a fixed value of $\alpha$) the second argument of \eqref{eq:symmetry_gap_simplified} is minimized for $w_2 = 2(1 - \alpha)/3 > 0$. Notice that this is a valid choice for $w_2$ since it implies that
\[
    w_1
    =
    1 - w_2/2 - \alpha
    =
    2(1 - \alpha)/3
    >
    0
    \enspace.
\]

Plugging the above value of $w_2$ (as a function of $\alpha$) into \eqref{eq:symmetry_gap_simplified}, we get a symmetry gap of
\[
    \max\bigg\{\alpha, \frac{(1 - \alpha + (1 - \alpha)/3)^2}{1 - \alpha + (1 - \alpha)}\bigg\}
    =
    \max\bigg\{\alpha, \frac{(4/3)^2 \cdot (1 - \alpha)}{2} \bigg\}
    =
    \max\{\alpha, 8(1 - \alpha)/9\}
    \enspace.
\]
Since the first argument of the maximum is increasing in $\alpha$ and the second argument is decreasing in $\alpha$, the value of the maximum is minimized when the two arguments are equal, i.e., when $\alpha = 8(1 - \alpha)/9$. Solving this equation, we get $\alpha = \frac{8}{17} \in (0, 1)$; and for this value of $\alpha$,
\[
    w_1
    =
    w_2
    =
    \frac{2(1 - \alpha)}{3}
    =
    \frac{2(1 - 8/17)}{3}
    =
    \frac{2 \cdot 9/17}{3}
    =
    \frac{6}{17}
    \enspace.
\]

\section*{AI Usage Disclosure}

The symmetric instance underlying our hardness results (i.e., the instance described in Section~\ref{sec:symmetric_instance}) was found by GPT-5.6-Sol during an extended automated exploration run through OpenAI's Codex harness. Along the way, this exploration considered several routes to improving the endpoint of Oveis Gharan \& Vondr{\'{a}}k~\cite{oveisgharan2011submodular}, including reweightings and multi-arm generalizations of their directed-hyperedge gadget, the cardinality variant of Qi~\cite{qi2024maximizing}, modular penalty and truncation repairs, and linear programs over all symmetric submodular objectives on small candidate geometries; most of these were ruled out by exact obstructions. The instance that eventually worked emerged late in this exploration.
  The authors then wrote the presentation herein from scratch without AI assistance, have checked all mathematical arguments, and take full responsibility for the paper's content.

\bibliographystyle{plainurl}
\bibliography{bib}

\end{document}